\documentclass[conference]{IEEEtran}

\IEEEoverridecommandlockouts

\usepackage{amsmath,amssymb,amsfonts}
\usepackage{amsthm}
\usepackage{nicefrac}
\newtheorem{theorem}{Theorem}
\newtheorem{corollary}{Corollary}
\usepackage{algorithmic}
\usepackage{graphicx}
\usepackage{textcomp}
\usepackage[table]{xcolor}
\usepackage{multirow}
\usepackage{makecell}
\newtheorem{proposition}{Proposition}
\newtheorem{lemma}{Lemma}

\def\BibTeX{{\rm B\kern-.05em{\sc i\kern-.025em b}\kern-.08em
    T\kern-.1667em\lower.7ex\hbox{E}\kern-.125emX}}

\usepackage{subfig}
\usepackage[nolist,nohyperlinks]{acronym}
\usepackage{siunitx}
\usepackage{cite}
\usepackage{todonotes}
\usepackage{microtype}
\usepackage{tikz}
\usetikzlibrary{decorations.pathreplacing}

\renewcommand{\baselinestretch}{.96}

\newcommand{\Conv}{%
    \mathop{\scalebox{1.5}{\raisebox{-0.2ex}{$\circledast$}}}
}

\definecolor{rbRed}{HTML}{FE433C}
\definecolor{rbPink}{HTML}{F31D64}
\definecolor{rbMagenta}{HTML}{A224AD}
\definecolor{rbViolet}{HTML}{6A38B3}
\definecolor{rbIndigo}{HTML}{3C50B1}
\definecolor{rbBlue}{HTML}{0095EF}

\acrodef{1-D}[1-D]{one-dimensional}
\acrodef{3-D}[3-D]{three-dimensional}
\acrodef{MC}[MC]{synthetic molecular communication}
\acrodef{MIGHT}[MIGHT]{Mixture of Inverse Gaussians for Hemodynamic Transport}
\acrodef{EM}[EM]{electromagnetic wave}
\acrodef{CGF}[CGF]{cumulant generating function}
\acrodef{CVS}[CVS]{cardiovascular system}
\acrodef{CIR}[CIR]{channel impulse response}
\acrodef{FPT}[FPT]{first passage time}
\acrodef{IoBNT}[IoBNT]{Internet of Bio-Nano Things}
\acrodef{SM}{signaling molecule}
\acrodefplural{SM}[SMs]{signaling molecules}
\acrodef{IG}[IG]{inverse Gaussian}
\acrodef{MGF}[MGF]{moment generating function}
\acrodef{PDF}[PDF]{probability density function}
\acrodef{RV}[RV]{random variable}
\acrodef{SNR}[SNR]{signal-to-noise ratio}
\acrodef{Tx}[Tx]{transmitter}
\acrodef{PBPK}{physiologically based pharmacokinetic}
\acrodef{Rx}[Rx]{receiver}
\acrodef{RMSE}{root mean square error}
\acrodefplural{RMSE}[RMSEs]{root mean square errors}
\acrodef{UCA}[UCA]{uniform concentration assumption}
\acrodef{VN}[VN]{vessel network}
\acrodef{wrt}[w.r.t.]{with respect to}
\acrodefplural{IG}[IGs]{inverse Gaussians}
\acrodefplural{PDF}[PDFs]{probability density functions}
\acrodefplural{RV}[RVs]{random variables}
\acrodefplural{SNR}[SNRs]{signal-to-noise ratios}
\acrodefplural{Rx}[Rxs]{receivers}
\acrodefplural{GFP}[GFPs]{green fluorescent proteins}
\acrodefplural{CIR}[CIRs]{channel impulse responses}
\acrodefplural{VN}[VNs]{vessel networks}
\acrodef{ROI}[ROI]{region of interest}
\acrodef{ICG}[ICG]{indocyanine green}

\begin{document}
\bstctlcite{IEEEexample:BSTcontrol}  

\title{Multiscale MIGHT: Modeling Local and Systemic Particle Transport in the Cardiovascular System%
}

\author{
\IEEEauthorblockN{\scalebox{1}{%
Theofilos Symeonidis\textsuperscript{\textdagger}, Timo Jakumeit\textsuperscript{\textdagger}, Fardad Vakilipoor,
Robert Schober, and Maximilian Schäfer
\thanks{\textsuperscript{\textdagger}Co-first authorship}
\thanks{This work was funded in part by the German Federal Ministry of Research, Technology and Space (BMFTR) through Project Internet of Bio-Nano-Things (IoBNT) -- grant number 16KIS1987, in part by the German Research Foundation (Deutsche Forschungsgemeinschaft, DFG) under GRK 2950 -- ProjectID 509922606 and under grant number SCHA 2350/2-1, in part by the European Union's Horizon Europe -- HORIZON-EIC-2024-PATHFINDEROPEN-01 under grant agreement Project N. 101185661, and in part by the Horizon Europe Marie Skodowska Curie Actions (MSCA)-UNITE under Project 101129618.}
}}\\[-0.4cm]
\IEEEauthorblockA{
Friedrich-Alexander-Universität Erlangen-Nürnberg (FAU),
Erlangen, Germany
}
}

\maketitle
\begin{abstract}
Analytical models of \ac{SM} transport in the \ac{CVS} force a choice between whole-body scope and vessel-level resolution: Pharmacokinetic models for systemic transport lump each organ and lack small-scale structure, whereas \ac{MC} channel models resolve branched \acp{VN} only for a single pass or capture recirculation only in topologically simple loops, and only partially include \ac{SM} degradation.
To close this gap, we propose \textit{Multiscale \acs{MIGHT}}, an extension of the \textit{\ac{MIGHT}} model that jointly captures branched topologies, closed-loop circulation, organ clearance, and chemical degradation of \acp{SM}.
First, we connect the ends of a branched \ac{VN} to form a closed loop with local and global \ac{SM} degradation, model the closed-loop channel impulse response for arbitrary \ac{Tx} and \ac{Rx} positions as a Laplace-domain feedback system, and approximate it in closed form by an inverse Gaussian mixture over multi-pass trajectories.
The model is validated against finite-element simulations in COMSOL.
Second, we set up a whole-body model, where an exemplary local \ac{VN}, resolved at vessel-level detail, is embedded into the remaining systemic circulation, represented by lumped organ blocks.
We simulate whole-body \ac{SM} transport for various \ac{Tx} sites and \ac{SM} species.
Our results show that, at an \ac{Rx} in a local \ac{VN}, only the initial pass produces a distinct peak in the \ac{SM} concentration, significantly shaped by the local vasculature in case of small \ac{Tx}--\ac{Rx} distances; subsequent passes arrive strongly dispersed and are shaped almost entirely by organ clearance and chemical degradation rather than by the local \ac{VN} topology.
\end{abstract}

\acresetall

\section{Introduction}
\label{sec:Introduction}
The \ac{CVS} distributes oxygen, nutrients, hormones, other endogenous \acp{SM}, and administered drugs throughout the body and is thus the primary domain for diagnostic and therapeutic in-body \ac{MC} applications.
Designing such applications requires analytical models of \ac{SM} transport in the \ac{CVS}.
However, to date, no such model jointly captures the three properties that arguably shape this transport most strongly:
(P1)~\textit{Branched topology:} The \ac{CVS} is a complex \ac{VN} with many parallel paths, so \acp{SM} reach an observation site along paths with widely differing travel times~\cite{Weiss2009}.
(P2)~\textit{Closed loop:} \acp{SM} recirculate, possibly along different paths in successive circulations~\cite{Weiss2009}.
(P3)~\textit{Degradation:} \acp{SM} are cleared locally by organs, e.g., liver and kidneys~\cite{Rowland2011PBPK}, and degraded globally, e.g., by enzymatic cleavage.

In a model, resolving the entire \ac{CVS} down to individual vessels is intractable.
For this reason, \textit{\ac{PBPK}} models abstract organs and tissues into well-mixed compartments~\cite{Rowland2011PBPK}.
Such models cover the whole body, including P2 and P3, but a lumped organ does not capture the internal structure and thus cannot predict the signal at a \ac{Rx} in a specific \ac{ROI}, i.e., a local \ac{VN}, thus missing P1.

In the \ac{MC} literature, in turn, P1 and P2 have been treated separately so far.
Early works modeled transport in single vessels, but the topology of a \ac{VN} strongly affects transport~\cite{Jakumeit2026}.
Hence, detailed \ac{VN} models have been proposed, cf., e.g., \cite{Chahibi2015}.
In this category, we previously proposed the \textit{\ac{MIGHT}} model~\cite{Jakumeit2026}, which offers a closed-form analytical description of \ac{SM} transport in \acp{VN}, but describes only a single pass through an open-ended \ac{VN} without degradation or recirculation, thus capturing P1 but missing P2 and P3.
Existing circulatory \ac{MC} models, on the other hand, capture the closed-loop characteristics (P2)~\cite{Vakilipoor2025} and partly degradation (P3)~\cite{Brand2025}, but represent the circulation as a topologically simple loop, missing~P1.

Existing models are thus whole-body but spatially unresolved~\cite{Rowland2011PBPK}, topology-aware but single-pass~\cite{Chahibi2015, Jakumeit2026}, or closed-loop but topologically simple~\cite{Vakilipoor2025, Brand2025}.
We close this gap by proposing the \textit{Multiscale \ac{MIGHT}} model, which retains lumped segments for the systemic circulation, but resolves the local circulation in an \ac{ROI}, e.g., an organ or the \ac{VN} around an \ac{Rx}, down to individual vessels, cf.~Fig.~\ref{fig:system_model}.
Multiscale \ac{MIGHT} is a closed-loop extension of \ac{MIGHT}~\cite{Jakumeit2026} with local and global degradation, and thereby captures P1--P3 jointly and accurately.

Our contributions are as follows:
(I)~We add local and global degradation to \ac{MIGHT} and derive the exact closed-loop \ac{CIR} for arbitrary \ac{Tx} and \ac{Rx} positions in the Laplace-domain, the expected number of \ac{Rx} crossings, the effective \ac{SM} elimination rate, and a closed-form \ac{IG}-mixture approximation in the time domain.
(II)~We validate the proposed model against numerical simulations in COMSOL Multiphysics\textsuperscript{\textregistered} and a convolution-based model for several closed-loop \acp{VN} with local and global degradation.
(III)~Using the proposed model, we derive a whole-body multiscale model comprised of lumped segments for systemic circulation and the \ac{ROI} represented as a resolved \ac{VN}. 
(IV)~We simulate \ac{SM} transport in the human body, reproduce known transport characteristics from compartmental models, and illustrate the interplay between local and systemic transport.

Sections~\ref{sec:ClosedLoopMIGHT}--\ref{sec:WholeBody} present the closed-loop model, its validation, the whole-body model topology, and predictions about \ac{SM} transport in the whole body, respectively.
Section~\ref{sec:Conclusion} concludes the paper.

\section{Closed-Loop MIGHT with Degradation}
\label{sec:ClosedLoopMIGHT}
In this section, we extend \ac{MIGHT}~\cite{Jakumeit2026} by adding degradation, extending it to closed-loop \acp{VN}, and deriving the closed-loop \ac{CIR} between arbitrarily placed \acp{Tx} and \acp{Rx}. We call the edges \emph{segments}; here, all segments are pipes, and Section~\ref{sec:Multiscale} adds lumped segments to represent organ passage.

\subsection{Closed-Loop Vessel Network and Assumptions}
\label{subsec:network}

\emph{Network.}
A \ac{VN} is a directed multigraph of pipes $p_i$, $i\in\{1,\ldots,E\}$, with length $l_i$, radius $r_i$, volume $V_i=\pi r_i^2l_i$, time-invariant flow rate $Q_i>0$, and mean velocity $u_i=Q_i/(\pi r_i^2)$, here with a single inlet pipe starting at $n_{\mathrm{in}}$ (inflow rate $Q$) and a single outlet node $n_{\mathrm{out}}$ (see~\cite[Sec.~II-A]{Jakumeit2026}). The $Q_i$-values follow from the hydraulic circuit model in~\cite[Sec.~II-C]{Jakumeit2026}. We construct a closed-loop \ac{VN} by joining $n_{\mathrm{out}}$ and $n_{\mathrm{in}}$ at the \emph{reference node} $n_0$.
Nodes with several outflow pipes are bifurcations $b_m$ with the set of outflow pipes $\mathcal O(b_m)$; $V=\sum_iV_i$ is the total volume of the \ac{VN}.

\emph{Tx \& Rx.}
The \ac{Tx} releases \acp{SM} at longitudinal position $z_{\mathrm{Tx}}$ in $p_q$, and the \ac{Rx} (length $l_{\mathrm{Rx}}$) is centered at $z_{\mathrm{Rx}}\neq z_{\mathrm{Tx}}$ in $p_\nu$. 

\emph{Passes \& crossings.}
In a closed-loop \ac{VN}, pass $0$ starts at the \ac{Tx}, pass $L\geq1$ starts at $n_{\mathrm{in}}$ after the $L$-th passage through $n_0$, and each pass either reaches $n_{\mathrm{out}}$ or terminates upon degradation. A pass from $n_{\mathrm{in}}$ to $n_{\mathrm{out}}$ constitutes a \emph{circulation}. 

\emph{Paths.}
As in~\cite[Sec.~II-A]{Jakumeit2026}, $\mathcal P(n_a,n_b)$ is the set of directed paths from $n_a$ to $n_b$ in the open-loop \ac{VN}. Path $P_k$ is defined as $P_k=\{p_i\mid i\in\mathcal E_k\}\cup\{b_m\mid m\in\mathcal B_k\}$, where $\mathcal E_k$ and $\mathcal B_k$ are the index sets of pipes and bifurcations in $P_k$, respectively.  
Moreover, $\mathcal S(p_i)$ and $\mathcal D(p_i)$ are the source and destination nodes of $p_i$. Conditioned on entering the first pipe of $P_k$, an \ac{SM} follows the remaining path with probability~\cite[Eq.~(20)]{Jakumeit2026}
\begin{align}
    \gamma_{P_k}
    \hspace*{-.5mm}=\hspace*{-.5mm}\prod_{p_j,b_m\in P_k;\,  p_j\in\mathcal O(b_m)}
     \frac{Q_j}{\sum_{p_{j'}\in\mathcal O(b_m)}Q_{j'}}, \hspace{1mm} \gamma_{P_k}\in [0,1].
    \label{eq:path_fraction}
\end{align}
To describe closed-loop \acp{VN}, we define the paths within one pass from the \ac{Tx} in $p_q$ to position $z$ in $p_\nu$ and to $n_{\mathrm{out}}$, and from $n_{\mathrm{in}}$ to position $z$ in $p_\nu$ and to $n_{\mathrm{out}}$ as follows
\begin{equation}
\begin{aligned}
    \mathcal P_{\mathrm{Tx},\nu}&=\{P_k\in\mathcal P(\mathcal S(p_q),\mathcal D(p_\nu))\mid p_q,p_\nu\in P_k\},\\
    \mathcal P_{\mathrm{Tx},\mathrm{out}}&=\{P_k\in\mathcal P(\mathcal S(p_q),n_{\mathrm{out}})\mid p_q\in P_k\},\\
    \mathcal P_{\mathrm{in},\nu}&=\{P_k\in\mathcal P(n_{\mathrm{in}},\mathcal D(p_\nu))\mid p_\nu\in P_k\},\\
    \mathcal P_{\mathrm{loop}}&=\mathcal P(n_{\mathrm{in}},n_{\mathrm{out}}).
\end{aligned}
    \label{eq:path_sets}
\end{equation}
By flow conservation, the path fractions in~\eqref{eq:path_fraction} sum to one over $\mathcal P_{\mathrm{Tx},\mathrm{out}}$ and $\mathcal P_{\mathrm{loop}}$ and to $Q_\nu/Q$ over $\mathcal P_{\mathrm{in},\nu}$.

To obtain the \ac{CIR} for closed-loop \acp{VN} with degradation, we make the following assumptions:
    
    \noindent(A1) \emph{Transport:} In a rigid pipe $p_i$, laminar, time-invariant flow yields \ac{1-D} advection-diffusion-driven transport with velocity $u_i$ and $D_{\mathrm{eff},i}=D+r_i^2u_i^2/(48D)$ in the Aris--Taylor regime, with diffusion coefficient $D$~\cite[Sec.~II]{Jakumeit2026}.
    
    \noindent(A2) \emph{Degradation:} \acp{SM} in pipe $p_i$ can degrade irreversibly with first-order reaction rate $\kappa_i\geq0$.

    \noindent(A3) \emph{Routing:} Bifurcations neither cause delay nor degradation, and \acp{SM} split in proportion to the outflow rates, independently of their history. Transit times are independent across segments and repeated traversals~\cite[Sec.~III-B]{Jakumeit2026}.
    
    \noindent(A4) \emph{Reception:} \acp{Rx} are transparent, the \ac{UCA} holds, and the diffusive flux at the \acp{Rx} is negligible compared to the advective flux~\cite{Jakumeit2026}.

\subsection{Pipe and Path Fluxes with Degradation}
\label{subsec:segments}

The \emph{pipe flux} $j_i(z,t)$ is the flux at $z\in(0,l_i]$ in pipe $p_i$ for one \ac{SM} entering at $z=0$ and $t=0$. Without degradation, it is the density of the \ac{SM}'s \ac{FPT} $T_i(z)$ at $z$, which is \ac{IG}-distributed~\cite[Theorem~1]{Jakumeit2026}, with the \ac{IG} density $f_{\mathrm{IG}}(t;\mu,\theta)=\frac{\mu}{\sqrt{2\pi\theta t^3}}\exp\bigl(-\frac{(t-\mu)^2}{2\theta t}\bigr)$ with mean $\mu$ and scale parameter $\theta$. With degradation in $p_i$, an \ac{SM} may degrade before reaching $z$, such that $T_i(z)$ becomes defective, i.e., its density is an exponentially tilted \ac{IG}, which is again a scaled \ac{IG}~\cite{Chhikara1989}.

\begin{lemma}[]
\label{lem:segment}
Under (A1) and (A2), an \ac{SM} entering $p_i$ at $z=0$ and $t=0$ reaches $z$ before degrading with probability
\begin{equation}
    S_i(z)=\exp\Bigl(\frac{z\,(u_i-\tilde u_i)}{2D_{\mathrm{eff},i}}\Bigr),\qquad
    \tilde u_i=\sqrt{u_i^2+4D_{\mathrm{eff},i}\kappa_i},
    \label{eq:segment_parameters}
\end{equation}
and, conditioned on this event, its \ac{FPT} is $T_i(z)\sim\mathrm{IG}(\mu_i(z),\theta_i)$ with $\mu_i(z)=z/\tilde u_i$ and $\theta_i=2D_{\mathrm{eff},i}/\tilde u_i^2$.

\noindent Hence, the pipe flux $j_i(z,t)$ in $\si{\per\second}$ at position $z$ follows as\footnote{For $\kappa_i=0$, $\tilde u_i=u_i$ and $S_i(z)=1$, which recovers~\cite[Theorem~1]{Jakumeit2026}.}
\begin{equation}
j_i(z,t)=S_i(z)\,f_{\mathrm{IG}}\bigl(t;\mu_i(z),\theta_i\bigr).
    \label{eq:segment_response}
\end{equation}
\end{lemma}
\begin{proof}
    The proof is provided in Appendix~\ref{app:proof_segment}.
\end{proof}

Exploiting (A2) and (A3), the \emph{path flux} of $P_k$ from $z_q$ in its first pipe $p_q$ to $z_w$ in its last pipe $p_w$ is~\cite[Eqs.~(9), (10)]{Jakumeit2026}
\begin{equation}
    \bar j_k(z_w,t;z_q)=j_q(l_q-z_q,t)*j_w(z_w,t)*\Conv_{i\in\mathcal E_k\setminus\{q,w\}}\!\!\!\!\!j_i(l_i,t),
    \label{eq:path_flux}
\end{equation}
where $*$ ($\Conv$) denotes convolution (over a set) in $t$, and $\bar j_k(z_w,t;z_q)=j_q(z_w-z_q,t)$ if $P_k$ consists of a single pipe, i.e., $q=w$, $z_w>z_q$.
Weighting the path fluxes~\eqref{eq:path_flux} of the paths in~\eqref{eq:path_sets} with their path fractions~\eqref{eq:path_fraction} as in~\cite[Eq.~(19)]{Jakumeit2026} yields the single-pass \acp{CIR} from the \ac{Tx} and from $n_{\mathrm{in}}$ to $z$ in $p_\nu$ and to $n_{\mathrm{out}}$, where paths to $n_{\mathrm{out}}$ end at $z_w=l_w$ in their last pipe $p_w$, as follows:
\begin{equation}
\begin{aligned}
    h_{\mathrm{Tx},\nu}(z,t)&=\textstyle\sum_{P_k\in\mathcal P_{\mathrm{Tx},\nu}}\gamma_{P_k}\,\bar j_k(z,t;z_{\mathrm{Tx}}),\\
    h_{\mathrm{Tx},\mathrm{out}}(t)&=\textstyle\sum_{P_k\in\mathcal P_{\mathrm{Tx},\mathrm{out}}}\gamma_{P_k}\,\bar j_k(l_w,t;z_{\mathrm{Tx}}),\\
    h_{\mathrm{in},\nu}(z,t)&=\textstyle\sum_{P_k\in\mathcal P_{\mathrm{in},\nu}}\gamma_{P_k}\,\bar j_k(z,t;0),\\
    h_{\mathrm{loop}}(t)&=\textstyle\sum_{P_k\in\mathcal P_{\mathrm{loop}}}\gamma_{P_k}\,\bar j_k(l_w,t;0),
\end{aligned}
    \label{eq:cir}
\end{equation}
which are zero for an empty path set. For $\kappa_i=0$, $h_{\mathrm{Tx},\nu}$ is the \ac{MIGHT} \ac{CIR}~\cite[Eq.~(19)]{Jakumeit2026}. Due to degradation and/or path bypasses, the \acp{CIR} in~\eqref{eq:cir} are generally defective densities, i.e., their integrals over $t$ are smaller than one.

\subsection{CIR for Closed-loop VNs}
\label{subsec:exact_response}

\begin{theorem}
\label{thm:closed_loop}
Let the segment fluxes $j_i\geq 0$ be possibly defective \ac{FPT} densities (see Lemma~\ref{lem:segment}) and assume that every circulation has strictly positive duration. Under (A2), (A3), the expected rate of an \ac{SM} released by the \ac{Tx} at $t=0$ crossing $z$ in $p_\nu$ is
\begin{equation}
\begin{aligned}
    h_{\mathrm{CL}}(z,t)&=\textstyle\sum_{L=0}^{\infty}h_{\mathrm{CL},L}(z,t),
    \qquad
    h_{\mathrm{CL},0}=h_{\mathrm{Tx},\nu},\\
    h_{\mathrm{CL},L}&=h_{\mathrm{Tx},\mathrm{out}}*h_{\mathrm{loop}}^{*(L-1)}*h_{\mathrm{in},\nu},
    \quad L\geq1,
\end{aligned}
    \label{eq:closed_loop_series}
\end{equation}
where $h_{\mathrm{CL},L}$ collects the crossings in pass $L$, $h^{*n}$ is the $n$-fold self-convolution, and $h^{*0}=\delta$ with Dirac delta $\delta$.
In the Laplace domain, \eqref{eq:closed_loop_series} is expressed as a feedback system (Fig.~\ref{fig:system_model}a)
\begin{equation}
    H_{\mathrm{CL}}(z,s)
    =H_{\mathrm{Tx},\nu}(z,s)
    +\frac{H_{\mathrm{Tx},\mathrm{out}}(s)\,H_{\mathrm{in},\nu}(z,s)}{1-H_{\mathrm{loop}}(s)}.
    \label{eq:closed_loop_laplace}
\end{equation}
Upper-case letters denote the Laplace transforms, e.g., $H_{\mathrm{CL}}(z,s)=\int_0^\infty h_{\mathrm{CL}}(z,t)\,e^{-st}\,\mathrm dt$, and $\operatorname{Re}\{s\}>0$. 
\end{theorem}
\begin{proof}
    The proof is provided in Appendix~\ref{app:proof_closed_loop}.
\end{proof}

\begin{corollary}[]
\label{cor:characteristics}
Exploiting assumptions (A1)-(A3), the following properties hold:

    \noindent(i) Let $S_{\mathrm{Tx},\nu}(z)$, $S_{\mathrm{Tx},\mathrm{out}}$, $S_{\mathrm{in},\nu}(z)$, and $S_{\mathrm{loop}}$ denote the masses $\int_0^\infty h(t)\,\mathrm dt$ of the respective \acp{CIR} in~\eqref{eq:cir}, i.e., the probabilities that the \ac{SM} arrives within the pass before degradation; e.g., $S_{\mathrm{loop}}$ is the survival probability per circulation. For $S_{\mathrm{loop}}<1$, the expected number of crossings of $z_{\mathrm{Rx}}$ per released \ac{SM} is
    \begin{align}
        \bar n_{\mathrm{Rx}}\hspace*{-.5mm}=\hspace*{-1.5mm}\int_0^\infty\hspace*{-4mm} h_{\mathrm{CL}}(z_{\mathrm{Rx}},t)\mathrm dt
        \hspace*{-.5mm}=\hspace*{-.5mm}S_{\mathrm{Tx},\nu}(z_{\mathrm{Rx}})\hspace*{-.5mm}+\hspace*{-.5mm}\tfrac{S_{\mathrm{Tx},\mathrm{out}}\,S_{\mathrm{in},\nu}(z_{\mathrm{Rx}})}{1-S_{\mathrm{loop}}}.
        \label{eq:expected_crossings}
    \end{align}

    \noindent (ii) For homogeneous degradation in a closed-loop \ac{VN}, i.e., $\kappa_i=\kappa, \forall p_i$, $h_{\mathrm{CL}}(z,t)$ is $e^{-\kappa t}$ times the \ac{CIR} for $\kappa=0$.

    \noindent (iii) For $\kappa_i=0,\forall p_i$, $h_{\mathrm{CL}}(z,t)\to Q_\nu/V$ for $t\to\infty$.
\end{corollary}
\begin{proof}
    The proof is provided in Appendix~\ref{app:proof_characteristics}.
\end{proof}

\subsection{Closed-Form IG-Mixture Approximation}
\label{subsec:closed_form}

In the following, we obtain a closed-form approximation of~\eqref{eq:closed_loop_laplace} through a moment matching approach, as in~\cite{Jakumeit2026}.  

\begin{proposition}[]
\label{lem:series}
Consider a chain of segment traversals $i\in\mathcal C$, e.g., along a path or along several passes, where a segment may recur, with independent transit times and fluxes $S_if_{\mathrm{IG}}(t;\mu_i,\theta_i)$.
The path flux is then the convolution of these fluxes (see \eqref{eq:path_flux}), i.e., $S_{\mathcal C}=\prod_iS_i$ times the density of the total transit time $T_{\mathcal C}$, which has the mean $\hat\mu_{\mathcal C}=\sum_i\mu_i$ and the variance $\hat\mu_{\mathcal C}\hat\theta_{\mathcal C}$ with $\hat\theta_{\mathcal C}=\sum_i\mu_i\theta_i/\hat\mu_{\mathcal C}$, but is generally not \ac{IG} distributed. Approximating $T_{\mathcal C}$ as in~\cite{Jakumeit2026}, gives
\begin{equation}
    \Conv_{i\in\mathcal C}S_if_{\mathrm{IG}}(t;\mu_i,\theta_i)
    \approx S_{\mathcal C}\,f_{\mathrm{IG}}\bigl(t;\hat\mu_{\mathcal C},\hat\theta_{\mathcal C}\bigr),
    \label{eq:series_rule}
\end{equation}
which is exact iff all $\theta_i$ are equal, and otherwise only underestimates the conditional skewness, cf. \cite[Theorem~3]{Jakumeit2026}.
\end{proposition}

\begin{proof}
    The proof is provided in Appendix~\ref{app:proof_series}.
\end{proof}

Applying proposition~\ref{lem:series} to~\eqref{eq:path_flux} gives the approximate path flux $\hat j_k=S_kf_{\mathrm{IG}}(t;\hat\mu_k,\hat\theta_k)$.
A \emph{trajectory} $\omega=(P_{k_0},\ldots,P_{k_L})\in\Omega_L$ is the sequence of paths taken in passes $0,\ldots,L$ up to a crossing in pass $L$, i.e., $P_{k_0}\in\mathcal P_{\mathrm{Tx},\nu}$ for $L=0$, and $P_{k_0}\in\mathcal P_{\mathrm{Tx},\mathrm{out}}$, $P_{k_\ell}\in\mathcal P_{\mathrm{loop}}$ for $1\leq\ell<L$, and $P_{k_L}\in\mathcal P_{\mathrm{in},\nu}$ for $L\geq1$.
Inserting~\eqref{eq:cir} into~\eqref{eq:closed_loop_series} gives $h_{\mathrm{CL},L}$ as a sum over $\Omega_L$ of $\prod_\ell\gamma_{P_{k_\ell}}$ multiplied with the convolution of the path fluxes along $\omega$, i.e., the flux of one long chain of segment traversals.

\begin{figure}[t]
    \centering
    \includegraphics[width=\linewidth]{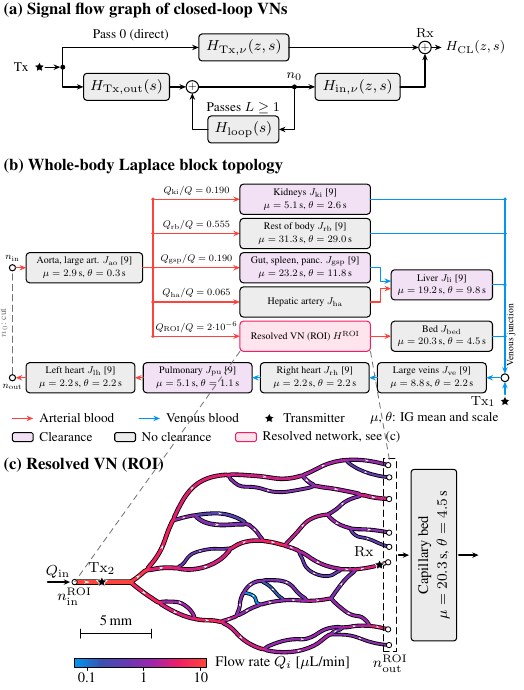}
    \caption{\textbf{Multiscale \ac{MIGHT} system model.} 
    (a) The four transfer functions of~\eqref{eq:closed_loop_laplace} for a given \ac{Tx}/\ac{Rx} and inlet/outlet placement. (b) Lumped organ segments with the resolved \ac{VN} embedded. (c) The \ac{ROI} in detail, draining into a lumped capillary bed.}
    \label{fig:system_model}
\end{figure}

\begin{proposition}[]
\label{prop:mixture}
Applying proposition~\ref{lem:series} to the trajectories in~\eqref{eq:closed_loop_series} yields the time domain closed-form \ac{CIR} of a closed-loop \ac{VN}
\begin{equation}
    h_{\mathrm{CL}}(z,t)\hspace*{-.7mm}\approx\hspace*{-.7mm}\hat h_{\mathrm{CL}}(z,t)
    \hspace*{-.5mm}=\hspace*{-1mm}\textstyle\sum_{L=0}^{\infty}\hspace*{-.5mm}\sum_{\omega\in\Omega_L}\hspace*{-1.5mm}W_\omega f_{\mathrm{IG}}\bigl(t;\hspace*{-.5mm}\hat\mu_\omega,\hspace*{-.5mm}\hat\theta_\omega\bigr),
    \label{eq:closed_loop_mixture}
\end{equation}
with $W_\omega\hspace*{-1mm}=\hspace*{-1mm}\prod_{\ell=0}^{L}\hspace*{-.5mm}\gamma_{P_{k_\ell}}\hspace*{-.5mm}S_{k_\ell}$, $\hat\mu_\omega\hspace*{-1mm}=\hspace*{-1mm}\sum_{\ell=0}^{L}\hspace*{-.5mm}\hat\mu_{k_\ell}$, $\hat\theta_\omega\hspace*{-1mm}=\hspace*{-1mm}\sum_{\ell=0}^{L}\hspace*{-.5mm}\hat\mu_{k_\ell}\hat\theta_{k_\ell}\hspace*{-.5mm}/\hspace*{-.5mm}\hat\mu_\omega$.
\end{proposition}
\begin{proof}
    The proof is provided in Appendix~\ref{app:proof_mixture}.
\end{proof}

Each trajectory in \eqref{eq:closed_loop_series} thus acts as one long path in the closed-loop \ac{VN}, and for $\kappa_i=0,\forall p_i$, the $L=0$ terms recover the closed-form \ac{MIGHT} \ac{CIR} in \cite{Jakumeit2026}.

\subsection{Received Signal}
\label{subsec:received_signal}

The expected number of \acp{SM} inside an \ac{Rx} at $z_{\mathrm{Rx}}$ in $p_\nu$ for a bolus release of $N$ \acp{SM} at the \ac{Tx} follows from \eqref{eq:closed_loop_mixture} as
\begin{equation}
    N_{\mathrm{obs,CL}}(t)\approx\tfrac{N\,l_{\mathrm{Rx}}}{u_\nu}\,h_{\mathrm{CL}}(z_{\mathrm{Rx}},t)\approx\tfrac{N\,l_{\mathrm{Rx}}}{u_\nu}\,\hat h_{\mathrm{CL}}(z_{\mathrm{Rx}},t).
    \label{eq:closed_loop_receiver_signal}
\end{equation}
The pipe and path fluxes, as well as the crossing-rate \ac{CIR} $h_{\mathrm{CL}}$, have units of $\mathrm{s}^{-1}$ per released \ac{SM}. Using (A4), the corresponding \ac{1-D} concentration response is approximately $h_{\mathrm{CL}}/u_\nu$, with units of $\mathrm{m}^{-1}$ per released \ac{SM} \cite{Jakumeit2026}. Alternatively, the received signal can be expressed as a concentration, in line with the clinical literature. Dividing~\eqref{eq:closed_loop_receiver_signal} by the observed volume $\pi r_\nu^2 l_\mathrm{Rx}$ cancels $l_\mathrm{Rx}$, so that for a bolus of solute mass $m_0=c_\mathrm{inj}V_\mathrm{inj}$, i.e., a volume $V_\mathrm{inj}$ of a solution of concentration $c_\mathrm{inj}$, released at the \ac{Tx}, the concentration at any position $z$ in the \ac{Rx} pipe $p_\nu$ follows as
\begin{equation}
        c_\mathrm{CL}(z,t)=\tfrac{m_0}{Q_\nu}\,h_\mathrm{CL}(z,t),\qquad Q_\nu=\pi r_\nu^2 u_\nu,
        \label{eq:closed_loop_concentration}
\end{equation}
in mass per volume, e.g., \si{\milli\gram\per\milli\liter} for $m_0$ in \si{\milli\gram} and $Q_\nu$ in \si{\milli\liter\per\second}. 
From Corollary~\ref{cor:characteristics}--(iii), without degradation, the received signal approaches the following uniform-mixing level
\begin{align}
    &N_{\mathrm{obs,CL}}(t)\xrightarrow[t\to\infty]{}\left(N\,l_{\mathrm{Rx}}\,\pi r_\nu^2\right)/V,
    &\!\!\!\!c_{\mathrm{CL}}(z,t)\xrightarrow[t\to\infty]{}m_0/V.
    \label{eq:closed_loop_equilibrium}
\end{align}
Replacing $m_0$ by the amount of substance $m_0/M$ or the \ac{SM} number $m_0 N_\mathrm{A}/M$, with molar mass $M$ and the Avogadro constant $N_\mathrm{A}$, gives the molar or number concentration instead.

\section{Numerical Validation}
\label{sec:NumericalValidation}
\begin{figure*}[t]
    \centering
    \includegraphics[width=\linewidth]{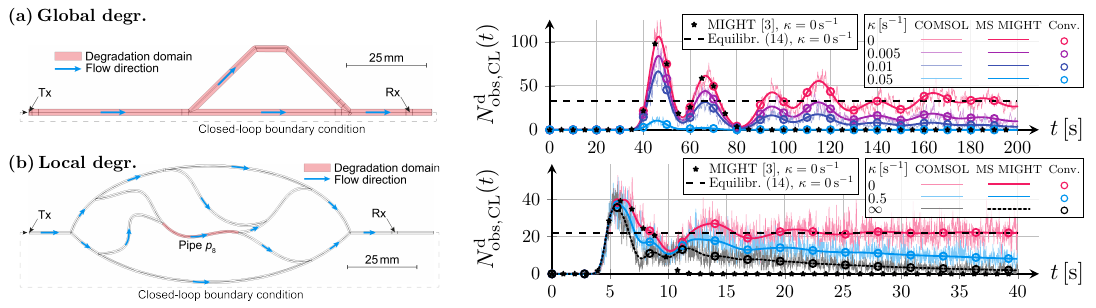}
    \caption{\textbf{Validation of Multiscale MIGHT.} Left: \acp{VN} and degradation domains. Right: $N_{\mathrm{obs,CL}}(t)$ in~(14) (thick), COMSOL (thin), convolution model (circles, Theorem~1); $D=\SI{1.461e-7}{\meter\squared\per\second}$; \textbf{a)} $Q=\SI{1.021e-8}{\meter\cubed\per\second}$, $N=3000$; \textbf{b)} $Q=\SI{5e-8}{\meter\cubed\per\second}$, $N=5000$.}
    \label{fig:validation}
\end{figure*}

We validate~\eqref{eq:closed_loop_receiver_signal} against particle-based finite-element simulations in COMSOL for the \acp{VN} in Fig.~\ref{fig:validation}, with the degradation rate $\kappa_i=\kappa$ in all pipes (global) or in pipe $p_8$ only (local).
The \ac{3-D} simulations couple the \textit{Laminar Flow} module (fully developed inflow, no slip, zero outlet pressure) with the \textit{Particle Tracing for Fluid Flow} module (Brownian motion, drag force).
A periodic boundary between outlet and inlet, continuous in flow rate and particle trajectories, closes the loop, with the \ac{Tx} (uniform random release over the cross-section) in the inlet and the \ac{Rx} in the outlet pipe.
First-order degradation with cumulative exposure over circulations is realized by a per-particle exponential clock that runs at rate $\kappa$ only inside pipes with degradation.
COMSOL signals are noisy due to Brownian motion and limited particle numbers.

Across all \acp{VN} and degradation rates in Fig.~\ref{fig:validation}, the proposed closed-loop model (thick solid, \eqref{eq:closed_loop_receiver_signal}) is in excellent agreement with the COMSOL results (thin solid) and the convolution-based response (circles, based on \eqref{eq:closed_loop_series}).
In Fig.~\ref{fig:validation}a), without degradation, the initial pass (star markers, \cite[Eq.~(24)]{Jakumeit2026}) shows one peak for each of the two \ac{Tx}--\ac{Rx} paths, and later passes superimpose combinations of their transit times.
Successive circulations remain distinguishable for long, and the signal approaches the equilibrium~\eqref{eq:closed_loop_equilibrium} only slowly.
With increasing global $\kappa$, the peaks decrease and the signal decays to zero for $t\rightarrow\infty$; for $\kappa=\SI{0.05}{\per\second}$, only the initial pass remains visible.

In Fig.~\ref{fig:validation}b), degradation is confined to $p_8$, whose four paths carry $\approx\SI{34}{\percent}$ of the particles per pass.
As these paths are among the slower ones, the initial peak amplitude is only moderately reduced and shifts slightly towards earlier times.
For $\kappa=\SI{0.5}{\per\second}$, $\SI{52}{\percent}$ of the particles entering $p_8$ survive it, i.e., $S_\mathrm{loop}=1-0.34\cdot(1-0.52)\approx0.84$, and the signal decays.
Local degradation acts only on particles entering its domain: for $\kappa\rightarrow\infty$, $p_8$ is a perfect sink, $S_\mathrm{loop}$ attains its lower bound $\num{0.66}$, and the resulting signal (dotted black) bounds the signal for any rate in $p_8$ from below; even then, a particle crosses the \ac{Rx} on average $(1-0.34)/0.34\approx1.9$ times, cf.~\eqref{eq:expected_crossings}.

\section{Multiscale Whole-Body Model}
\label{sec:Multiscale}
In the following, we extend the models for closed-loop \acp{VN} from Section~\ref{sec:ClosedLoopMIGHT} to a whole-body model comprising an exemplary resolved \ac{VN} embedded in the systemic circulation, modeled by lumped segments representing (groups of) organs.

\subsection{Lumped Segments}
\label{subsec:lumped}

The $i$-th \emph{lumped segment} has a flow rate $Q_i>0$ and a volume $V_i>0$ but no internal positions, such that we write its flux as $j_i(t)$ and its transfer function as $J_i(s)$. It can represent, e.g., a heart chamber, the pulmonary circulation, an organ, or a capillary bed, cf.\ Fig.\ \ref{fig:system_model}b).
Studies of systemic transit times in mammals~\cite{Weiss2009} have shown that the transit time through the $i$-th lumped segment can be modeled as \ac{IG}-distributed.
Without degradation, we use the mean $\mu_i=V_i/Q_i$ (central volume principle~\cite{Meier1954}) and scale $\theta_i=\mathrm{CV}_i^2\mu_i$, with coefficient of variation $\mathrm{CV}_i$.
Values for $\mu_i$ and $\mathrm{CV}_i$ for a reference adult male, weighing $\SI{73}{\kilo\gram}$, were taken from~\cite{Valentin2002} where available and otherwise assumed or fitted to~\cite{Weiss2009}, cf.\ Fig.\ \ref{fig:system_model}b).
With \ac{SM}-dependent degradation rate $\kappa_i$ (e.g., dependent on organ extraction ratio $\varepsilon_i$), an \ac{SM} with residence time $t$ survives in the $i$-th lumped segment with probability $e^{-\kappa_it}$, and its flux is
\begin{equation}
j_i(t)=S_i\,f_{\mathrm{IG}}\bigl(t;\mu_i/\rho_i,\theta_i/\rho_i^2\bigr),
    \label{eq:lumped_degradation}
\end{equation}
with $\rho_i=\sqrt{1+2\theta_i\kappa_i}$ and survival probability $S_i=\exp\bigl(\frac{\mu_i}{\theta_i}(1-\rho_i)\bigr)$.
Organ clearance is given by the extraction ratio $\varepsilon_i$, i.e., the fraction of arriving \acp{SM} removed in one pass~\cite{Leevy1962HepaticFlow}.
Requiring $S_i=1-\varepsilon_i$, $\kappa_i$ is then obtained as
\begin{equation}
    \kappa_i= \frac{1}{2\theta_i}\left[ 
    \bigl(1-\tfrac{\theta_i}{\mu_i}
    \ln(1-\varepsilon_i)\bigr)^2-1\right].
    \label{eq:app_extraction_ratio}
\end{equation}
In some cases, several parallel lumped segments $i\in\mathcal{G}$ without clearance are represented by an equivalent lumped segment indexed by $\mathcal{G}$, e.g., "\textit{Rest of body}" in Fig.\ \ref{fig:system_model}b).
There, it holds that $Q_{\mathcal G}=\sum_{i\in\mathcal G}Q_i$, $V_{\mathcal G}=\sum_{i\in\mathcal G}V_i$, $\mu_{\mathcal G}=V_{\mathcal G}/Q_{\mathcal G}$, and
\begin{equation}
    \mu_\mathcal{G}\theta_\mathcal{G} = \textstyle\sum_{i\in\mathcal{G}}\frac{Q_i}{Q_\mathcal{G}}\mu_i\theta_i + \textstyle\sum_{i\in\mathcal{G}}\frac{Q_i}{Q_\mathcal{G}}\left( \mu_i - \mu_\mathcal{G}\right)^2,
    \label{eq:segment_aggregation}
\end{equation}
which matches the first two moments of the equivalent lumped segment to those of the flow-weighted mixture of its constituent segments.
Segments with clearance are kept separate. 
The aggregation in~\eqref{eq:segment_aggregation} is optional: the constituent organs can instead be represented as individual parallel lumped segments with their respective flow rates, volumes, and transit-time parameters, without modifying the closed-loop formulation.
In Fig.\ \ref{fig:system_model}b), the hepatic artery is an exception and represents a direct link with no delay or spread, i.e., flux $\delta(t)$ with $J_{\mathrm{ha}}=1$, necessary to account for the fractions of flows entering the liver. Fig.\ \ref{fig:transit_time_distributions} illustrates the \ac{IG} transit time distributions for data based on~\cite{Valentin2002} as well as the multipath transit time distribution for the resolved \ac{VN} in Fig.\ \ref{fig:system_model}c).

\subsection{Whole-Body Network Topology}
\label{subsec:whole_body_graph}

The considered whole-body network topology is shown in Fig.\ \ref{fig:system_model}b). 
In the central body part, large veins, right heart, pulmonary vasculature, left heart, and aorta and large arteries are in series.
Kidneys, gut, spleen, and pancreas, the resolved \ac{VN} with a concatenated capillary bed, and the rest of the body (brain, coronary circulation, skeletal muscle, skin, adipose, tissue, skeleton, bronchial circulation, and other tissues) follow in a parallel structure behind the central body part.
Flow rate fractions $Q_i/Q$, cardiac output $Q=\SI{6.5}{\liter\per\minute}$, and total blood volume $V=\SI{5.3}{\liter}$ are taken from~\cite[Secs.~7.4, 7.5]{Valentin2002}.
$Q_\mathrm{ROI}/Q=2\cdot 10^{-6}$ is chosen as an educated guess for a small local region of the body ($V_\mathrm{ROI}$ and $Q_\mathrm{ROI}$ are subtracted from the "\textit{Rest of the body}" segment).
The gut, spleen, and pancreas as well as the hepatic artery lead into the liver, which is connected in series.
In Fig.\ \ref{fig:system_model}b), we place $\mathrm{Tx}_1$ at the venous junction, where deoxygenated blood from the organs flows back into the central segment.
A second $\mathrm{Tx}_2$ is placed at the beginning of the resolved \ac{VN}, a \ac{Rx} is placed towards the end of the resolved \ac{VN}, cf.\ Fig.\ \ref{fig:system_model}c).
The cut node $n_0$ is placed between the inlet $n_\mathrm{in}$ at the aortic root and the outlet $n_\mathrm{out}$ at the left-ventricular outflow, Fig.\ \ref{fig:system_model}b).

In this work, the exemplary resolved \ac{VN} resembles anastomotic vasculature with radii $r_i\in[80,165]\,\SI{}{\micro\meter}$ and vessel lengths in the $\SI{}{\milli\meter}$ range, but does not correspond to a specific tissue/organ in the body. This is left for future work.
The exemplary resolved \ac{VN} comprises a single inlet $n_\mathrm{in}^\mathrm{ROI}$ and several outlets (here jointly denoted by $n_\mathrm{out}^\mathrm{ROI}$) that drain into a lumped capillary bed with manually chosen mean and spread, which in turn drains into the venous junction, cf. Fig.\ \ref{fig:system_model}c).

\subsection{Whole-Body Network Model}
\label{subsec:block_assembly}

Using the transfer function framework (Section~\ref{sec:ClosedLoopMIGHT}), each segment in Fig.\ \ref{fig:system_model}b) enters through its Laplace transform, i.e., $J_i(s)$ of~\eqref{eq:lumped_degradation} for lumped segments and $J_i(z,s)$ of Lemma~\ref{lem:segment} for the \ac{ROI}.
As path fractions~\eqref{eq:path_fraction} are products of local flow fractions and path fluxes~\eqref{eq:path_flux} are products of transfer functions, the path sum~\eqref{eq:cir} factors along the graph: segments in series multiply, parallel branches are weighted by their flow fractions. The transfer function $H_\mathrm{loop}(s)$, cf.\ \eqref{eq:cir}, in Fig.\ \ref{fig:system_model}a) then becomes
\begin{equation}
\begin{aligned}
    H_{\mathrm{loop}}(s)&=J_{\mathrm{ao}}(s)\,H_{\mathrm{sys}}(s)\,J_{\mathrm{ret}}(s),\\
    H_{\mathrm{sys}}(s)&=\tfrac{Q_{\mathrm{ki}}}{Q}J_{\mathrm{ki}}(s)+\tfrac{Q_{\mathrm{rb}}}{Q}J_{\mathrm{rb}}(s)+\tfrac{Q_{\mathrm{ROI}}}{Q}H^{\mathrm{ROI}}_{\mathrm{in},\mathrm{out}}(s)\,J_{\mathrm{bed}}(s)\\
    &\quad+\tfrac{J_{\mathrm{li}}(s)}{Q}\bigl(Q_{\mathrm{ha}}+Q_{\mathrm{gsp}}J_{\mathrm{gsp}}(s)\bigr),
\end{aligned}
    \label{eq:whole_body_loop}
\end{equation}
with the return chain $J_{\mathrm{ret}}=J_{\mathrm{ve}}J_{\mathrm{rh}}J_{\mathrm{pu}}J_{\mathrm{lh}}$ of the large veins (ve), right heart (rh), pulmonary circulation (pu), and left heart (lh), the aorta and large arteries (ao), kidneys (ki), rest of the body (rb), liver (li), hepatic artery (ha), and gut, spleen, and pancreas (gsp), and $Q_{\mathrm{ki}}+Q_{\mathrm{rb}}+Q_{\mathrm{ha}}+Q_{\mathrm{gsp}}+Q_{\mathrm{ROI}}=Q$.
While the above transfer functions correspond to lumped segments without spatial resolution, $H^{\mathrm{ROI}}_{\mathrm{in},\mathrm{out}}$ is the spatially resolved path sum $\sum_{P_k}\gamma_{P_k}\bar J_k(l_w,s;0)$ through the \ac{ROI}, and $H^{\mathrm{ROI}}_{\mathrm{in},\nu}$, $H^{\mathrm{ROI}}_{\mathrm{Tx},\nu}$, and $H^{\mathrm{ROI}}_{\mathrm{Tx},\mathrm{out}}$ are the single-pass \acp{CIR}~\eqref{eq:cir} of the \ac{ROI} alone, with path fractions relative to $Q_{\mathrm{ROI}}$.
As the \ac{Rx} at $z$ in pipe $p_\nu$ of the \ac{ROI} is reached from $n_{\mathrm{in}}$ only through ao, $H_{\mathrm{in},\nu}=J_{\mathrm{ao}}\frac{Q_{\mathrm{ROI}}}{Q}H^{\mathrm{ROI}}_{\mathrm{in},\nu}$.
For $\mathrm{Tx}_1$, pass $0$ traverses only the return chain, i.e., $H_{\mathrm{Tx},\nu}=0$ and $H_{\mathrm{Tx},\mathrm{out}}=J_{\mathrm{ret}}$, and~\eqref{eq:closed_loop_laplace} gives
\begin{equation}
    H_{\mathrm{CL},1}(z,s)=\frac{J_{\mathrm{ret}}(s)\,J_{\mathrm{ao}}(s)\,\frac{Q_{\mathrm{ROI}}}{Q}\,H^{\mathrm{ROI}}_{\mathrm{in},\nu}(z,s)}{1-H_{\mathrm{loop}}(s)}.
    \label{eq:whole_body_cl}
\end{equation}
For $\mathrm{Tx}_2$ in the \ac{ROI}, pass $0$ reaches the \ac{Rx} within the \ac{ROI} or leaves it through the capillary bed into the return chain, i.e., $H_{\mathrm{Tx},\nu}=H^{\mathrm{ROI}}_{\mathrm{Tx},\nu}$ and $H_{\mathrm{Tx},\mathrm{out}}=H^{\mathrm{ROI}}_{\mathrm{Tx},\mathrm{out}}J_{\mathrm{bed}}J_{\mathrm{ret}}$, such that
\begin{equation}
    H_{\mathrm{CL},2}\hspace*{-.5mm}(z,\hspace*{-.5mm}s)\hspace*{-.7mm}=\hspace*{-.9mm}H^{\mathrm{ROI}}_{\mathrm{Tx},\nu}\hspace*{-.5mm}(z,\hspace*{-.5mm}s)\hspace*{-.4mm}+\hspace*{-.7mm}H^{\mathrm{ROI}}_{\mathrm{Tx},\mathrm{out}}\hspace*{-.5mm}(s)J_{\mathrm{bed}}\hspace*{-.5mm}(s)H_{\mathrm{CL},1}\hspace*{-.5mm}(z,\hspace*{-.5mm}s),
    \label{eq:local_cl}
\end{equation}
where the second subscript labels the \ac{Tx}.
In the following, only a single \ac{Tx} is active in each scenario; multiple simultaneously active \acp{Tx} would be captured by superposition as in~\cite{Jakumeit2026}.

Since each lumped segment has a scaled-\ac{IG} transform, while the resolved \ac{ROI} contributes finite sums of products of such transforms, ~\eqref{eq:whole_body_cl} and~\eqref{eq:local_cl} are explicit in $s$ and are evaluated exactly by numerical Laplace inversion, while the derivation of a closed-form approximation is left for future work. 

\section{Whole-Body Simulation Results}
\label{sec:WholeBody}

\begin{figure}[t]
    \centering
    \includegraphics[width=\linewidth]{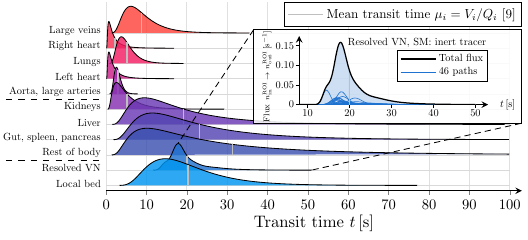}
    \caption{\textbf{Whole-body transit times.} Lumped segments are \acp{IG}~\cite{Weiss2009} with parameters from~\cite{Valentin2002}. Only the resolved \ac{VN} is multipath, with fluxes obtained from \ac{MIGHT}. Distributions normalized for comparability.}
    \label{fig:transit_time_distributions}
\end{figure}

Below, we study the propagation of two \ac{SM} species predicted by the whole-body model in Section~\ref{sec:Multiscale}, with $D=\SI{5e-8}{\meter\squared\per\second}$.
The first is \ac{ICG}, injected as a bolus intravenously at $\mathrm{Tx}_1$ (or intra-arterially at $\mathrm{Tx}_2$).
\ac{ICG} is cleared only in the liver, with a hepatic extraction ratio of $\varepsilon_\mathrm{li}=0.7285$~\cite[Tab.~IX]{Leevy1962HepaticFlow}, i.e., $\kappa_\mathrm{li}=\SI{0.0906}{\per\second}$ via~\eqref{eq:app_extraction_ratio}, and does not degrade chemically \textit{in vivo} on the considered time scales~\cite{Leevy1962HepaticFlow}, i.e., its global degradation rate is zero.
The second is a generic tracer, for which we vary the organ extraction ratios and global degradation rates; tracers that are neither cleared nor degraded are denoted as \textit{inert tracers}.
Concentrations refer to whole blood, i.e., we neglect the plasma confinement of protein-bound \acp{SM}.

Fig.~\ref{fig:tx1_vs_tx2_signals} shows the concentration at the center position $z_\mathrm{Rx}$ of the \ac{Rx} for a bolus of $m_0=\SI{1}{\micro\gram}$ of \ac{ICG} or inert tracer injected at either $\mathrm{Tx}_1$ (venous junction) or $\mathrm{Tx}_2$ (inlet pipe of the local \ac{VN}), cf.~Figs.~\ref{fig:system_model}b) and~\ref{fig:system_model}c).
As $\mathrm{Tx}_2$ is close to $z_\mathrm{Rx}$, the initial peak of the $\mathrm{Tx}_2$ curves resolves the multipath structure of the local \ac{VN}.
This multimodality of $c_\mathrm{CL}$ is lost in all subsequent circulations, which are dominated by systemic dispersion, and it is absent even in the first peak for $\mathrm{Tx}_1$, where the much larger distance to $z_\mathrm{Rx}$ disperses the signal before reaching $z_\mathrm{Rx}$.
In the long term, the inert tracer approaches the equilibrium in~\eqref{eq:closed_loop_equilibrium}, whereas \ac{ICG} decays to zero due to hepatic clearance.
Note, also, that the \ac{ICG} curves for $\mathrm{Tx}_1$ and $\mathrm{Tx}_2$ converge for large $t$, indicating that the long-term distribution is governed by clearance rather than by the injection site.

\begin{figure}
    \centering
    \includegraphics[width=\linewidth]{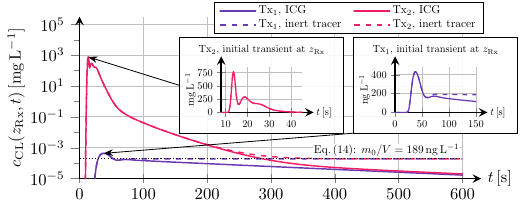}
    \caption{\textbf{Concentration at $z_\mathrm{Rx}$ for bolus injection at $\mathrm{Tx}_1$ or $\mathrm{Tx}_2$.} A bolus injection of $m_0=\SI{1}{\micro\gram}$ \ac{ICG} or inert tracer is simulated.}
    \label{fig:tx1_vs_tx2_signals}
\end{figure}

\begin{figure}
    \centering
    \includegraphics[width=\linewidth]{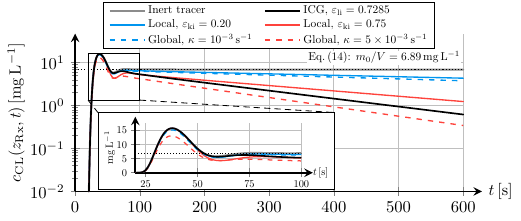}
    \caption{\textbf{Long-term systemic degradation of different \acp{SM}.} Varying renal/hepatic extraction ratios and global degradation rates are shown. A bolus injection of dose $m_0=\SI{36.5}{\milli\gram}$ at $\mathrm{Tx}_1$ is simulated.}
    \label{fig:long_term_decay}
\end{figure}

Fig.~\ref{fig:long_term_decay} shows the concentration at $z_\mathrm{Rx}$ for a $m_0=\SI{36.5}{\milli\gram}$ bolus at $\mathrm{Tx}_1$, for \ac{ICG} (black), an inert tracer (grey), and tracers with renal extraction (solid) or global degradation (dashed).
Since the direct path from $\mathrm{Tx}_1$ to $z_\mathrm{Rx}$ passes no clearing organ, the locally cleared species share the inert tracer's first peak and deviate only later, whereas global degradation attenuates the signal from the outset.
Note, also, that \ac{ICG} and the strongly cleared renal tracer (red solid) have nearly equal extraction ratios, yet their curves cross at $\approx\SI{110}{\second}$.
This is because the kidney has much faster transit times than the liver, cf.\ Fig.\ \ref{fig:transit_time_distributions}.
Initially, the signal affected by renal clearance (red solid) therefore lies below the \ac{ICG} curve. 
Only after $\approx\SI{110}{\second}$, the liver's larger share of the cardiac output compensates for the slower transit times and the \ac{ICG} concentration falls below the renal clearance signal (red solid).
Moreover, comparing the dashed lines to their solid counterparts, it can be observed that small global degradation rates have a much larger impact than large local degradation rates in the long term.

In summary, Figs.\ \ref{fig:tx1_vs_tx2_signals} and \ref{fig:long_term_decay} qualitatively reproduce the systemic concentration profile measured in mammals~\cite{Weiss2009, Vakilipoor2025} and predicted by \ac{PBPK} models~\cite{Rowland2011PBPK}, while additionally enabling the investigation of the impact of an \ac{ROI} around the \ac{Rx}.

\section{Conclusion}
\label{sec:Conclusion}
In this work, we introduced \textit{Multiscale \ac{MIGHT}}, a closed-loop extension of the \ac{MIGHT} model in~\cite{Jakumeit2026} with local and global \ac{SM} degradation.
We formulated the model as a Laplace-domain feedback system and derived the exact \ac{CIR} for arbitrary \ac{Tx} and \ac{Rx} positions in closed-loop \acp{VN}, as well as a closed-form \ac{IG}-mixture \ac{CIR} approximation, both of which agree closely with COMSOL simulations.
Based on this, we proposed a multiscale whole-body model, comprising a vessel-resolved \ac{ROI} and lumped segments for the remaining organs.
Unlike pharmacokinetic models, this model can predict the concentration at a specific \ac{Rx} position within an \ac{ROI}, while also capturing the systemic circulation.

In future work, we aim to validate the whole-body model with clinical data of systemic distribution, extend it with more detailed models for selected organs and tissues, and employ it to aid the design of diagnostic and therapeutic \ac{MC} applications.

\appendices
\section{Proof of Lemma~\ref{lem:segment}}
\label{app:proof_segment}

Under (A1) and (A2), the cross-sectionally averaged concentration $c_i(z,t)$ of non-degraded \acp{SM} in pipe $p_i$, caused by a single \ac{SM} entering at $z=0$ at time $t=0$, obeys the \ac{1-D} advection--diffusion equation with a first-order sink,
\begin{equation}
    \partial_tc_i(z,t)=D_{\mathrm{eff},i}\,\partial_{zz}c_i(z,t)-u_i\,\partial_zc_i(z,t)-\kappa_i\,c_i(z,t),
    \label{eq:advection_diffusion_degradation}
\end{equation}
with the initial condition $c_i(z,0)=0$, the influx boundary condition at the pipe inlet~\cite[Eq.~(35)]{Jakumeit2026}
\begin{equation}
    -D_{\mathrm{eff},i}\,\partial_zc_i(0,t)+u_i\,c_i(0,t)=\delta(t),
    \label{eq:app_influx}
\end{equation}
and $c_i(z,t)\to0$ for $z\to\infty$.
Compared with~\cite[Eq.~(34)]{Jakumeit2026}, degradation adds only the sink $-\kappa_ic_i$, since, by (A2), degraded \acp{SM} do not contribute to $c_i$ and are not observed.
The pipe flux follows as~\cite[Eq.~(40)]{Jakumeit2026}
\begin{equation}
    j_i(z,t)=u_i\,c_i(z,t)-D_{\mathrm{eff},i}\,\partial_zc_i(z,t).
    \label{eq:app_flux_def}
\end{equation}
Applying the one-sided Laplace transform with respect to $t$, denoted by $\mathcal L\{\cdot\}$ with complex frequency $s$, and exploiting the initial condition $c_i(z,0)=0$, transforms~\eqref{eq:advection_diffusion_degradation} into
\begin{equation}
    D_{\mathrm{eff},i}\,\partial_{zz}C_i(z,s)-u_i\,\partial_zC_i(z,s)-(s+\kappa_i)\,C_i(z,s)=0,
    \label{eq:app_laplace_ode}
\end{equation}
where $C_i(z,s)=\mathcal L\{c_i(z,t)\}$; here and in the following, upper-case letters denote Laplace transforms.
Hence,~\eqref{eq:app_laplace_ode} coincides with its degradation-free counterpart~\cite[Eq.~(36)]{Jakumeit2026} with $s$ replaced by $s+\kappa_i$, i.e., degradation acts as a shift in the Laplace domain.

As in~\cite[Eq.~(37)]{Jakumeit2026}, we use the ansatz $C_i(z,s)=A_i(s)\,e^{\lambda_i(s)z}$. Inserting it into~\eqref{eq:app_laplace_ode} gives $D_{\mathrm{eff},i}\lambda_i^2-u_i\lambda_i-(s+\kappa_i)=0$, i.e.,
\begin{equation}
    \lambda_i(s)=\frac{u_i-\sqrt{u_i^2+4D_{\mathrm{eff},i}(s+\kappa_i)}}{2D_{\mathrm{eff},i}},
    \label{eq:app_lambda}
\end{equation}
where the minus sign is required for $c_i(z,t)\to0$ as $z\to\infty$, cf.~\cite[Eq.~(38)]{Jakumeit2026}.
The Laplace transform of the influx boundary condition~\eqref{eq:app_influx}, $-D_{\mathrm{eff},i}\,\partial_zC_i(0,s)+u_i\,C_i(0,s)=1$, yields $A_i(s)=1/(u_i-D_{\mathrm{eff},i}\lambda_i(s))$, cf.~\cite[Eq.~(39)]{Jakumeit2026}.
Transforming the flux~\eqref{eq:app_flux_def} and inserting $C_i(z,s)$ finally gives
\begin{align}
    J_i(z,s) &= u_iC_i(z,s)-D_{\mathrm{eff},i}\,\partial_zC_i(z,s)\\
    &=\bigl(u_i-D_{\mathrm{eff},i}\lambda_i(s)\bigr)A_i(s)\,e^{\lambda_i(s)z}
    =e^{\lambda_i(s)z}.
    \label{eq:app_pipe_laplace}
\end{align}

With the effective velocity $\tilde u_i=\sqrt{u_i^2+4D_{\mathrm{eff},i}\kappa_i}$ from Lemma~\ref{lem:segment}, the radicand in~\eqref{eq:app_lambda} becomes $u_i^2+4D_{\mathrm{eff},i}(s+\kappa_i)=\tilde u_i^2+4D_{\mathrm{eff},i}s$, such that by adding and subtracting $\tilde u_i$ in the numerator, \eqref{eq:app_pipe_laplace} can be rewritten as
\begin{equation}
    J_i(z,s)
    =\underbrace{\exp\Bigl(\frac{z\,(u_i-\tilde u_i)}{2D_{\mathrm{eff},i}}\Bigr)}_{=\,S_i(z)}\,
    \exp\Bigl(\frac{z\,\bigl(\tilde u_i-\sqrt{\tilde u_i^2+4D_{\mathrm{eff},i}s}\bigr)}{2D_{\mathrm{eff},i}}\Bigr),
    \label{eq:app_split}
\end{equation}
where the first factor does not depend on $s$ and equals $S_i(z)$ in~\eqref{eq:segment_parameters}.
The second factor is the degradation-free pipe flux~\cite[Eq.~(41)]{Jakumeit2026} with $u_i$ replaced by $\tilde u_i$ and $D_{\mathrm{eff},i}$ unchanged.
With the Laplace transform of the \ac{IG} density~\cite{Chhikara1989},
\begin{equation}
    F_{\mathrm{IG}}(s;\mu,\theta)=\mathcal L\{f_{\mathrm{IG}}(t;\mu,\theta)\}=\exp\Bigl(\frac{\mu}{\theta}\bigl(1-\sqrt{1+2\theta s}\bigr)\Bigr),
    \label{eq:app_FIG}
\end{equation}
and $\mu=\mu_i(z)=z/\tilde u_i$ and $\theta=\theta_i=2D_{\mathrm{eff},i}/\tilde u_i^2$, we have $\mu/\theta=z\tilde u_i/(2D_{\mathrm{eff},i})$ and $\sqrt{1+2\theta s}=\sqrt{\tilde u_i^2+4D_{\mathrm{eff},i}s}/\tilde u_i$, such that the second factor in~\eqref{eq:app_split} equals $F_{\mathrm{IG}}(s;\mu_i(z),\theta_i)$, i.e.,
\begin{equation}
    J_i(z,s)=S_i(z)\,F_{\mathrm{IG}}\bigl(s;\mu_i(z),\theta_i\bigr).
    \label{eq:app_pipe_laplace_ig}
\end{equation}

As $S_i(z)$ does not depend on $s$, the inverse Laplace transform $\mathcal L^{-1}\{\cdot\}$ of~\eqref{eq:app_pipe_laplace_ig} acts only on $F_{\mathrm{IG}}$, and~\cite[Eq.~(42)]{Jakumeit2026} with $u_i$ replaced by $\tilde u_i$ gives
\begin{equation}
\begin{aligned}
    j_i(z,t)&=S_i(z)\,\frac{z}{\sqrt{4\pi D_{\mathrm{eff},i}t^3}}\exp\Bigl(-\frac{(z-\tilde u_it)^2}{4D_{\mathrm{eff},i}t}\Bigr)\\
    &=S_i(z)\,\frac{\mu_i(z)}{\sqrt{2\pi\theta_it^3}}\exp\Bigl(-\frac{(t-\mu_i(z))^2}{2\theta_it}\Bigr)\\
    &=S_i(z)\,f_{\mathrm{IG}}\bigl(t;\mu_i(z),\theta_i\bigr),
\end{aligned}
    \label{eq:app_time_domain}
\end{equation}
where the second step uses $\mu_i(z)=z/\tilde u_i$ and $\theta_i=2D_{\mathrm{eff},i}/\tilde u_i^2$, cf.~\cite[Eq.~(43)]{Jakumeit2026}, which gives~\eqref{eq:segment_response}.

Since $F_{\mathrm{IG}}(0;\cdot)=1$,~\eqref{eq:app_pipe_laplace_ig} gives $S_i(z)=J_i(z,0)=\int_0^\infty j_i(z,t)\,\mathrm dt$, i.e., $S_i(z)$ is the probability that the \ac{SM} reaches $z$ before it degrades, and $j_i(z,t)/S_i(z)=f_{\mathrm{IG}}(t;\mu_i(z),\theta_i)$ is the density of its \ac{FPT} conditioned on this event. This concludes the proof. \qed

\textit{Remark:} 
Equivalently, the first line of~\eqref{eq:app_time_domain} equals the degradation-free flux~\cite[Eq.~(42)]{Jakumeit2026} multiplied by $e^{-\kappa_it}$, since, by assumption (A2), each \ac{SM} survives its time $t$ in the pipe with probability $e^{-\kappa_it}$. As $\tilde u_i\geq u_i$, the surviving \acp{SM} are faster and less dispersed than without degradation, i.e., $\mu_i(z)\leq z/u_i$ and $\theta_i\leq2D_{\mathrm{eff},i}/u_i^2$, while $S_i(z)\leq1$ decreases exponentially in $z$.
Moreover, as the exponent $\lambda_i(s)z$ in~\eqref{eq:app_pipe_laplace} is linear in $z$, $J_i(z,s)\,J_i(z'-z,s)=J_i(z',s)$ for $0<z<z'$; correspondingly, $\mu_i(z)$ and $\ln S_i(z)$ are linear in $z$, while $\theta_i$ does not depend on $z$.
Hence, splitting a pipe at any position is exact, and a release at $z$ is equivalent to a release at the entrance of a pipe of length $l_i-z$, cf.~\cite[App.~A]{Jakumeit2026}, which justifies the distances in~\eqref{eq:path_flux}.

\section{Proof of Theorem~\ref{thm:closed_loop}}
\label{app:proof_closed_loop}
\label{app:reference_node}

The proof only required the fluxes $j_i$ and flow rates $Q_i$ of the segments; hence, it holds for pipes as well as for the lumped segments introduced in Section~\ref{sec:Multiscale}.

Consider a single \ac{SM} released by the \ac{Tx} at $t=0$, with its passes as defined in Section~\ref{subsec:network}.
Let $T_L$ denote the time at which the \ac{SM} crosses $z$ in $p_\nu$ during pass $L\geq0$, where $T_L=\infty$ if it does not cross $z$ in this pass, e.g., because it degrades before or takes a path that bypasses $p_\nu$.
As the open-loop \ac{VN} is acyclic, each path visits $p_\nu$ at most once, such that each pass contains at most one crossing. 
Hence, the number of crossings of $z$ in $[0,t]$ is $N_z(t)=\sum_{L=0}^{\infty}\mathbb 1\{T_L\leq t\}$, with the indicator function $\mathbb 1\{\cdot\}$, and its expected value follows as 
\begin{equation}
    \mathbb E\bigl[N_z(t)\bigr]=\sum_{L=0}^{\infty}\Pr\{T_L\leq t\}, 
    \label{eq:app_crossings}
\end{equation}
where $\mathbb E[\cdot]$ denotes expectation and we used $\mathbb E[\mathbb 1\{A\}]=\Pr\{A\}$.
As all terms are nonnegative, differentiating~\eqref{eq:app_crossings} termwise with respect to $t$ gives the expected crossing rate $h_{\mathrm{CL}}(z,t)=\sum_{L\geq0}h_{\mathrm{CL},L}(z,t)$, where $h_{\mathrm{CL},L}(z,\cdot)$ is the defective density of $T_L$, whose mass, i.e., its integral over $t$, is $\Pr\{T_L<\infty\}$ and generally smaller than one; here and below, the mass of a nonnegative function is its integral over $t$.
Now, $h_{\mathrm{CL},L}$ for $L=0$ and $L\geq1$ remains to be determined.

In pass~$0$, the \ac{SM} starts at $z_{\mathrm{Tx}}$ in $p_q$ and can cross $z$ only along a path $P_k\in\mathcal P_{\mathrm{Tx},\nu}$, cf.~\eqref{eq:path_sets}, which it follows with probability $\gamma_{P_k}$ from~\eqref{eq:path_fraction} by (A3); along all other paths, it bypasses $p_\nu$.
Along $P_k$, the transit times through the traversed segments are independent by (A3), and the survival probability given these transit times factorizes over the segments by (A2).
Hence, the defective density of the arrival time at $z$ along $P_k$ is the convolution of the segment fluxes, evaluated at the traveled distances, i.e., the path flux $\bar j_k(z,t;z_{\mathrm{Tx}})$ in~\eqref{eq:path_flux}, where each pipe flux is the defective \ac{FPT} density in Lemma~\ref{lem:segment}.
The law of total probability over the paths then yields
\begin{equation}
    h_{\mathrm{CL},0}(z,t)=\sum_{P_k\in\mathcal P_{\mathrm{Tx},\nu}}\gamma_{P_k}\,\bar j_k(z,t;z_{\mathrm{Tx}})=h_{\mathrm{Tx},\nu}(z,t),
    \label{eq:app_pass0}
\end{equation}
cf.~\eqref{eq:cir}.
The same holds for a \ac{Tx} at the entrance of a segment $p_q$, e.g., a lumped segment in Section~\ref{sec:Multiscale}, with $z_{\mathrm{Tx}}=0$ and $p_q$ traversed entirely.

A crossing in pass $L\geq1$ requires three stages without degradation: (a)~reaching $n_{\mathrm{out}}$ from the \ac{Tx} in pass~$0$, after time $\tau_0$; (b)~completing circulations $\ell=1,\ldots,L-1$, after time $\tau_\ell$; and (c)~reaching $z$ from $n_{\mathrm{in}}$ in pass~$L$, after time $\tau_L$. As in pass~$0$, the law of total probability over the paths in $\mathcal P_{\mathrm{Tx},\mathrm{out}}$, $\mathcal P_{\mathrm{loop}}$, and $\mathcal P_{\mathrm{in},\nu}$, respectively, shows that $\tau_0$, $\tau_\ell$, and $\tau_L$ have the defective densities $h_{\mathrm{Tx},\mathrm{out}}$, $h_{\mathrm{loop}}$, and $h_{\mathrm{in},\nu}(z,\cdot)$ in~\eqref{eq:cir}, where each is $\infty$ if the \ac{SM} degrades or bypasses the respective end point.

By (A3), the routing decisions and transit times of different passes are independent, and by (A2), the survival probability given all transit times factorizes over the passes.
Hence, the stage times are independent, and
\begin{equation}
    T_L=\tau_0+\sum_{\ell=1}^{L-1}\tau_\ell+\tau_L
    \label{eq:app_TL}
\end{equation}
has the defective density
\begin{equation}
h_{\mathrm{CL},L}(z,t)=h_{\mathrm{Tx},\mathrm{out}}(t)*h_{\mathrm{loop}}^{*(L-1)}(t)*h_{\mathrm{in},\nu}(z,t),\qquad L\geq1,
    \label{eq:app_passL}
\end{equation}
as the density of a sum of independent random variables is the convolution of their densities, with $h_{\mathrm{loop}}^{*0}=\delta$ for $L=1$. Together with~\eqref{eq:app_pass0}, this gives~\eqref{eq:closed_loop_series}.

In the Laplace domain the convolution theorem turns~\eqref{eq:app_pass0} and~\eqref{eq:app_passL} into $H_{\mathrm{CL},0}(z,s) = H_{\mathrm{Tx},\nu}(z,s)$ and
\begin{align}
    &H_{\mathrm{CL},L}(z,s)=H_{\mathrm{Tx},\mathrm{out}}(s)\,H_{\mathrm{loop}}^{L-1}(s)\,H_{\mathrm{in},\nu}(z,s), &L\geq1.
    \label{eq:app_passL_laplace}
\end{align}
The loop density $h_{\mathrm{loop}}\geq0$ has the mass $S_{\mathrm{loop}}=H_{\mathrm{loop}}(0)=\sum_{P_k\in\mathcal P_{\mathrm{loop}}}\gamma_{P_k}\int_0^\infty\bar j_k(l_w,t;0)\,\mathrm dt\leq\sum_{P_k\in\mathcal P_{\mathrm{loop}}}\gamma_{P_k}=1$, since every path flux has a mass of at most one and the path fractions sum to one by the conservation of flow. 
Moreover, as every circulation has strictly positive duration, $h_{\mathrm{loop}}$ has no mass at $t=0$, such that $e^{-at}<1$ almost everywhere on its support for $a>0$.
Hence, for $a=\operatorname{Re}\{s\}>0$,
\begin{equation}
    |H_{\mathrm{loop}}(s)|\leq\int_0^\infty h_{\mathrm{loop}}(t)\,e^{-at}\,\mathrm dt=H_{\mathrm{loop}}(a)<1,
    \label{eq:app_loop_bound}
\end{equation}
where the strict inequality holds for $S_{\mathrm{loop}}>0$ because $H_{\mathrm{loop}}(a)<H_{\mathrm{loop}}(0)\leq1$, and trivially for $S_{\mathrm{loop}}=0$.
As all terms of~\eqref{eq:closed_loop_series} are nonnegative and~\eqref{eq:app_loop_bound} ensures that the resulting series converges, the series can be transformed termwise yielding
\begin{align}
    H_{\mathrm{CL}}(z,s)&=H_{\mathrm{Tx},\nu}(z,s)\\&+H_{\mathrm{Tx},\mathrm{out}}(s)\,H_{\mathrm{in},\nu}(z,s)\sum_{L=1}^{\infty}H_{\mathrm{loop}}^{L-1}(s)\\
    &=H_{\mathrm{Tx},\nu}(z,s)+\frac{H_{\mathrm{Tx},\mathrm{out}}(s)\,H_{\mathrm{in},\nu}(z,s)}{1-H_{\mathrm{loop}}(s)},
    \label{eq:app_geometric}
\end{align}
i.e.,~\eqref{eq:closed_loop_laplace}.
This concludes the proof.
\qed

\section{Proof of Corollary~\ref{cor:characteristics}}
\label{app:proof_characteristics}

As in Appendix~\ref{app:proof_closed_loop}, the proof requires only the fluxes and flow rates of segments; hence, it also holds for the lumped segments in Section~\ref{sec:Multiscale}.

\emph{Part (i).} For nonnegative functions $f$ and $g$, the mass of their convolution is the product of their masses, $\int_0^\infty(f*g)(t)\,\mathrm dt=\int_0^\infty f(t)\,\mathrm dt\int_0^\infty g(t)\,\mathrm dt$, cf.\ Appendix~\ref{app:proof_closed_loop}.
Hence, by~\eqref{eq:closed_loop_series}, $h_{\mathrm{CL},0}$ has the mass $S_{\mathrm{Tx},\nu}(z)$ of $h_{\mathrm{Tx},\nu}$, and $h_{\mathrm{CL},L}$ has the mass $S_{\mathrm{Tx},\mathrm{out}}S_{\mathrm{loop}}^{L-1}S_{\mathrm{in},\nu}(z)$ for $L\geq1$, i.e., the product of the masses of $h_{\mathrm{Tx},\mathrm{out}}$, $h_{\mathrm{loop}}$, and $h_{\mathrm{in},\nu}$.

As all terms are nonnegative, integration and summation can be exchanged, and for $S_{\mathrm{loop}}<1$, the geometric series converges:
\begin{equation}
\begin{aligned}
    \bar n_{\mathrm{Rx}}&=\sum_{L=0}^{\infty}\int_0^\infty h_{\mathrm{CL},L}(z_{\mathrm{Rx}},t)\,\mathrm dt\\
    &=S_{\mathrm{Tx},\nu}(z_{\mathrm{Rx}})+S_{\mathrm{Tx},\mathrm{out}}\,S_{\mathrm{in},\nu}(z_{\mathrm{Rx}})\sum_{L=1}^{\infty}S_{\mathrm{loop}}^{L-1},
\end{aligned}
    \label{eq:app_crossings_sum}
\end{equation}
which gives~\eqref{eq:expected_crossings}.

\emph{Part (ii).} For $\kappa_i=\kappa$ in all segments, each segment flux is $e^{-\kappa t}$ times its counterpart for $\kappa=0$, i.e., the \ac{SM} survives time $t$ in the segment with probability $e^{-\kappa t}$. 
The factor $e^{-\kappa t}$ distributes over convolutions as follows
\begin{align}
    \bigl(e^{-\kappa t}f\bigr)*\bigl(e^{-\kappa t}g\bigr)(t)
    &=\int_0^t e^{-\kappa\tau}f(\tau)\,e^{-\kappa(t-\tau)}g(t-\tau)\,\mathrm d\tau\nonumber\\
    &=e^{-\kappa t}\,(f*g)(t),
    \label{eq:app_tilt_conv}
\end{align}
and over sums, while the path fractions $\gamma_{P_k}$ do not depend on $\kappa$. Hence, the path fluxes~\eqref{eq:path_flux}, the single-pass \acp{CIR}~\eqref{eq:cir}, and the pass densities in~\eqref{eq:closed_loop_series} are all $e^{-\kappa t}$ times their counterparts for $\kappa=0$, and so is their sum $h_{\mathrm{CL}}(z,t)$. 

\emph{Part (iii).} Let $\kappa_i=0$ for all $i$, and define $\varphi=h_{\mathrm{Tx},\mathrm{out}}*h_{\mathrm{in},\nu}(z,\cdot)$ and $\psi=h_{\mathrm{CL}}(z,\cdot)-h_{\mathrm{Tx},\nu}(z,\cdot)$.
By~\eqref{eq:closed_loop_series}, $\psi=\sum_{L\geq1}\varphi*h_{\mathrm{loop}}^{*(L-1)}$, which satisfies the renewal equation
\begin{equation}
    \psi=\varphi+h_{\mathrm{loop}}*\psi.
    \label{eq:app_renewal}
\end{equation}
Without degradation, $h_{\mathrm{loop}}$ is a proper density, since its mass is $\sum_{P_k\in\mathcal P_{\mathrm{loop}}}\gamma_{P_k}=1$ by the conservation of flow. 
Moreover, every \ac{IG} flux decays as $t^{-3/2}e^{-t/(2\theta_i)}$ for $t\to\infty$, such that $h_{\mathrm{loop}}$ has a finite mean $m_{\mathrm{loop}}=\int_0^\infty t\,h_{\mathrm{loop}}(t)\,\mathrm dt$, $\varphi$ is continuous and decays exponentially, hence is directly integrable, and $h_{\mathrm{Tx},\nu}(z,t)\to0$.
The key renewal theorem~\cite{Asmussen2003} applied to~\eqref{eq:app_renewal} then gives
\begin{equation}
    \lim_{t\to\infty}h_{\mathrm{CL}}(z,t)=\lim_{t\to\infty}\psi(t)=\frac{1}{m_{\mathrm{loop}}}\int_0^\infty\varphi(t)\,\mathrm dt.
    \label{eq:app_renewal_limit}
\end{equation}
By Part~(i), $\int_0^\infty\varphi(t)\,\mathrm dt=S_{\mathrm{Tx},\mathrm{out}}\,S_{\mathrm{in},\nu}(z)$, where, without degradation, $S_{\mathrm{Tx},\mathrm{out}}=\sum_{P_k\in\mathcal P_{\mathrm{Tx},\mathrm{out}}}\gamma_{P_k}=1$ and $S_{\mathrm{in},\nu}(z)=\sum_{P_k\in\mathcal P_{\mathrm{in},\nu}}\gamma_{P_k}=Q_\nu/Q$ by the conservation of flow.
Moreover, the mean of a convolution is the sum of the means, such that the mean of the path flux of $P_k$ is the sum of the mean transit times $\mu_i=V_i/Q_i$ of its segments, which holds for pipes by Lemma~\ref{lem:segment} with $\kappa_i=0$, i.e., $\mu_i=l_i/u_i=V_i/Q_i$, and for lumped segments by definition (see Section~\ref{sec:Multiscale}). Since the circulations traversing segment $p_i$, i.e., the paths $P_k\in\mathcal P_{\mathrm{loop}}$ with $i\in\mathcal E_k$, carry the fraction $Q_i/Q$ of the flow, exchanging the sums over paths and segments $i$ gives the mean circulation time as follows
\begin{align}
    m_{\mathrm{loop}}
    &=\sum_{P_k\in\mathcal P_{\mathrm{loop}}}\gamma_{P_k}\sum_{i\in\mathcal E_k}\mu_i
    =\sum_{i}\mu_i\sum_{\substack{P_k\in\mathcal P_{\mathrm{loop}}\\ i\in\mathcal E_k}}\gamma_{P_k}\nonumber\\
    &=\sum_i\frac{V_i}{Q_i}\,\frac{Q_i}{Q}=\frac{V}{Q},
    \label{eq:app_mean_circulation}
\end{align}
Inserting both into~\eqref{eq:app_renewal_limit} gives $h_{\mathrm{CL}}(z,t)\to(Q_\nu/Q)/(V/Q)=Q_\nu/V$ for $t\to\infty$, which concludes the proof.
\qed

\section{Proof of Proposition~\ref{lem:series}}
\label{app:proof_series}

Let $T_i$ denote the transit time of traversal $i\in\mathcal C$.
By Lemma~\ref{lem:segment}, the flux $j_i(t)=S_if_{\mathrm{IG}}(t;\mu_i,\theta_i)$ is the defective density of the \ac{SM} completing traversal $i$ after time $t$ without degrading.
As the transit times are independent by (A3) and, by (A2), survival in one traversal depends only on the time spent in it, the joint defective density of completing all traversals after times $t_i$ factorizes,
\begin{equation}
    \prod_{i\in\mathcal C}S_if_{\mathrm{IG}}(t_i;\mu_i,\theta_i)
    =S_{\mathcal C}\prod_{i\in\mathcal C}f_{\mathrm{IG}}(t_i;\mu_i,\theta_i).
    \label{eq:app_chain_joint}
\end{equation}
Hence, the chain is completed before degradation with probability $S_{\mathcal C}=\prod_iS_i$, and, conditioned on completing the chain, the $T_i$ are independent $\mathrm{IG}(\mu_i,\theta_i)$ variables, such that the left-hand side of~\eqref{eq:series_rule} is $S_{\mathcal C}$ times the density of $T_{\mathcal C}=\sum_iT_i$.

The cumulants of an \ac{IG} variable follow from the derivatives of $\ln F_{\mathrm{IG}}(s;\mu,\theta)=\frac{\mu}{\theta}\bigl(1-\sqrt{1+2\theta s}\bigr)$ at $s=0$ as $\mu$, $\mu\theta$, and $3\mu\theta^2$~\cite[Eq.~(50)]{Jakumeit2026}.
As cumulants of independent variables add, $T_{\mathcal C}$ has mean, variance, and third cumulant
\begin{equation}
    \textstyle\sum_i\mu_i=\hat\mu_{\mathcal C},\qquad
    \sum_i\mu_i\theta_i=\hat\mu_{\mathcal C}\hat\theta_{\mathcal C},\qquad
    3\sum_i\mu_i\theta_i^2,
    \label{eq:app_chain_cumulants}
\end{equation}
respectively, whereas the approximation $\mathrm{IG}(\hat\mu_{\mathcal C},\hat\theta_{\mathcal C})$ has the cumulants $\hat\mu_{\mathcal C}$, $\hat\mu_{\mathcal C}\hat\theta_{\mathcal C}$, and $3\hat\mu_{\mathcal C}\hat\theta_{\mathcal C}^2$.
Thus,~\eqref{eq:series_rule} preserves the mass $S_{\mathcal C}$ and the conditional mean and variance.

For the third cumulant, let $w_i=\mu_i/\hat\mu_{\mathcal C}$, such that $\sum_iw_i=1$ and $\sum_iw_i\theta_i=\hat\theta_{\mathcal C}$.
Jensen's inequality then gives
\begin{equation}
    3\textstyle\sum_i\mu_i\theta_i^2=3\hat\mu_{\mathcal C}\sum_iw_i\theta_i^2\geq3\hat\mu_{\mathcal C}\hat\theta_{\mathcal C}^2,
    \label{eq:app_jensen}
\end{equation}
with equality if and only if all $\theta_i$ are equal, and the difference is $3\hat\mu_{\mathcal C}\operatorname{Var}_w(\theta_i)$ with $\operatorname{Var}_w(\theta_i)=\sum_iw_i(\theta_i-\hat\theta_{\mathcal C})^2$.
As the variances coincide, dividing by $(\hat\mu_{\mathcal C}\hat\theta_{\mathcal C})^{3/2}$ gives the skewness error
\begin{equation}
    \Delta\xi_{\mathcal C}=3\,\frac{\operatorname{Var}_w(\theta_i)}{\hat\theta_{\mathcal C}^2}\sqrt{\frac{2}{\mathrm{Pe}_{\mathcal C}}}\geq0,\qquad
    \mathrm{Pe}_{\mathcal C}=\frac{2\hat\mu_{\mathcal C}}{\hat\theta_{\mathcal C}},
    \label{eq:app_skewness}
\end{equation}
i.e., the approximation underestimates the conditional skewness, by an amount that vanishes for large P\'eclet numbers, cf.~\cite[Theorem~3]{Jakumeit2026}.

If all $\theta_i=\theta$, then $\prod_iF_{\mathrm{IG}}(s;\mu_i,\theta)=\exp\bigl(\frac{\hat\mu_{\mathcal C}}{\theta}(1-\sqrt{1+2\theta s})\bigr)=F_{\mathrm{IG}}(s;\hat\mu_{\mathcal C},\theta)$, i.e.,~\eqref{eq:series_rule} is exact; otherwise, $\Delta\xi_{\mathcal C}>0$, and it is not.
Traversals with $\mu_i=0$, e.g., delay-free connections with flux $\delta(t)$, leave the convolution unchanged and are omitted from $\mathcal C$.
This concludes the proof.
\qed

\section{Proof of Proposition~\ref{prop:mixture}}
\label{app:proof_mixture}

Substituting the single-pass \acp{CIR} from~\eqref{eq:cir} into~\eqref{eq:closed_loop_series} and expanding the convolutions over the path sums gives, for every $L\geq0$,
\begin{equation}
    h_{\mathrm{CL},L}(z,t)=\sum_{\omega\in\Omega_L}\Bigl(\prod_{\ell=0}^{L}\gamma_{P_{k_\ell}}\Bigr)\bigl(\bar j_{k_0}*\cdots*\bar j_{k_L}\bigr)(t),
    \label{eq:app_trajectory_expansion}
\end{equation}
where $\bar j_{k_0}$ starts at $z_{\mathrm{Tx}}$, $\bar j_{k_L}$ ends at $z$, and all other path fluxes start at $n_{\mathrm{in}}$ and end at $n_{\mathrm{out}}$, cf.~\eqref{eq:cir}.
By~\eqref{eq:path_flux}, each path flux is itself the convolution of its segment fluxes, evaluated at the distances traveled, such that each term in~\eqref{eq:app_trajectory_expansion} is the flux of the chain $\mathcal C_\omega$ of all segment traversals of $\omega$, in which a segment may recur in different passes.

Applying Proposition~\ref{lem:series} to $\mathcal C_\omega$ replaces this flux by $S_{\mathcal C_\omega}f_{\mathrm{IG}}(t;\hat\mu_{\mathcal C_\omega},\hat\theta_{\mathcal C_\omega})$.
As the survival probabilities multiply and the conditional means $\mu_i$ and variances $\mu_i\theta_i$ add, grouping the traversals by passes gives
\begin{equation}
\begin{aligned}
    S_{\mathcal C_\omega}&=\textstyle\prod_{\ell=0}^{L}S_{k_\ell},\qquad
    \hat\mu_{\mathcal C_\omega}=\sum_{\ell=0}^{L}\hat\mu_{k_\ell}=\hat\mu_\omega,\\
    \hat\mu_{\mathcal C_\omega}\hat\theta_{\mathcal C_\omega}&=\textstyle\sum_{\ell=0}^{L}\hat\mu_{k_\ell}\hat\theta_{k_\ell}=\hat\mu_\omega\hat\theta_\omega,
\end{aligned}
    \label{eq:app_trajectory_params}
\end{equation}
i.e., matching per trajectory yields the same parameters as first matching per path, cf.~\eqref{eq:series_rule}.
With $W_\omega=\prod_\ell\gamma_{P_{k_\ell}}S_{k_\ell}$, each term in~\eqref{eq:app_trajectory_expansion} thus becomes $W_\omega f_{\mathrm{IG}}(t;\hat\mu_\omega,\hat\theta_\omega)$, and summing over $\omega$ and $L$ gives~\eqref{eq:closed_loop_mixture}.
This concludes the proof.
\qed

\textit{Remark:}
By Proposition~\ref{lem:series}, each term preserves its mass and its conditional mean and variance. As $\Omega_L$ is finite and the mass, mean, and variance of a finite mixture are determined by those of its components, the approximation preserves the mass and, for every $L$ with nonzero mass, the conditional mean and variance of $h_{\mathrm{CL},L}$. In particular, $\hat h_{\mathrm{CL}}$ and $h_{\mathrm{CL}}$ have the same mass in every pass, such that~\eqref{eq:expected_crossings} holds for $\hat h_{\mathrm{CL}}$ if $S_{\mathrm{loop}}<1$. For uniform degradation, Corollary~\ref{cor:characteristics}(ii) generally holds for $\hat h_{\mathrm{CL}}$ only approximately, since exponential tilting and moment matching need not commute. They commute for $\kappa=0$ or when all traversed \ac{IG} segments have the same degradation-free scale.

\bibliography{bibliography}

\end{document}